\documentclass{article}

\usepackage{template}

\usepackage[utf8]{inputenc} 
\usepackage[T1]{fontenc}    
\usepackage{hyperref}       
\usepackage{url}            
\usepackage{booktabs}       
\usepackage{amsfonts}       
\usepackage{nicefrac}       
\usepackage{microtype}      
\usepackage{lipsum}

\usepackage{amssymb}
\usepackage{graphicx} 
\usepackage{subcaption}
\usepackage{array} 
\usepackage{verbatim} 
\usepackage{amsmath}
\usepackage{amsfonts}
\usepackage{amssymb}
\usepackage{amsthm}
\usepackage{soul}
\usepackage{color}
\usepackage{booktabs}

\usepackage[lined,algonl,boxed]{algorithm2e}

\newtheorem{theorem}{Theorem}

\newtheorem{definition}{Definition}

\title{Governance of Social Welfare in Networked Markets}

\author{
  MohammadAmin Fazli \\
  Department of Computer Engineering\\
  Sharif University of Technology\\
  Tehran, Iran \\
  \texttt{fazli@sharif.edu} \\
   \And
 Alireza Amanihamedani \\
  Department of Computer Engineering\\
  Sharif University of Technology\\
  Tehran, Iran \\
  \texttt{aramani@ce.sharif.edu} \\
}

\begin{document}
\maketitle

\begin{abstract}
This paper aims to investigate how a central authority (e.g. a government) can increase social welfare in a network of markets and firms. In these networks, modeled using a bipartite graph, firms compete with each other \textit{\`a la} Cournot. Each firm can supply homogeneous goods in markets which it has access to. The central authority may take different policies for its aim. In this paper, we assume that the government has a budget by which it can supply some goods and inject them into various markets. We discuss how the central authority can best allocate its budget for the distribution of goods to maximize social welfare. We show that the solution is highly dependent on the structure of the network.  Then, using the network's structural features, we present a heuristic algorithm for our target problem. Finally, we compare the performance of our algorithm with other heuristics with experimentation on real datasets. 
\end{abstract}


\section{Introduction}
Cournot Competition in single-market setting has been vastly studied in the literature \cite{kreps1983quantity,fang2015managing,tramontana2015local,constantatos2019accommodation}. In this oligopolistic model, each firm decides the quantity of the homogenous good it is willing to supply to the market. Then, according to the inverse demand function, the market-clearing price is determined based on the aggregate supply in the market. However, with the emergence of diverse and complicated economic scenarios, single-market models are inadequate for studying reality. In many settings, firms can compete in different markets, whether or not the good is identical in those markets. Typically, this situation is modeled using a bipartite graph in which one side of nodes represents firms and the other side depicts various markets. Each market is characterized by an inverse demand function. Multi-market competition is found abundantly in industries such as natural gas, water, electricity, airlines, cement, healthcare, etc \cite{skantze2000stochastic,heggestad1978multi,gilbert2017multi,villena2017dynamics}.

One question that arises naturally in the presence of strategic agents is the means by which it is possible to raise welfare measures \cite{just1979welfare,bateman2006aggregation}. One such measure is social welfare, which aims to capture the aggregate well-being in the economic environment \cite{sen2018collective,clarke1997managerial}. In this case, it is typically the government that seeks and has the responsibility of higher social welfare. While there have been many studies on interactions between firms and equilibria in networked markets (See e.g. \cite{bertrand,shayan,mamtagh}), to the best of our knowledge, there is little work on how to govern and control social welfare in networked markets. The prevalence of networked markets in real-life experiences motivates us to study social welfare in this model. Our paper takes one step forward towards this objective.

We consider a limited intervention budget for the government in the pursuit of higher social welfares. Therefore, we assume that the government is able to have a small amount of supply into every market. This small intervention setting enables us to use some techniques for the estimation of social welfare in terms of government's supplies. The simple structure of the approximation leads to a strategy for the government. However, it is good to note that the actions taken by the government, are specified by the structure of the network. After all, this structure, along with the consumers' behaviors, are sufficient statistics to specify the market.

\subsection{Related Works}
Our work is in essence related to several groups of papers. First, there have been many attempts in studying the strategic behavior of firms and equilibria in the competition \cite{nava2009quantity,shayan,anshelevich2015price,pang2017efficiency}. One such study, which has been our first step-stone, is done by Bimpikis el al. \cite{shayan}, where they present a "characterization of the production quantities at the unique equilibrium of the resulting game for any given network", in terms of supply paths in the network. Furthermore, they introduce the \textit{price-impact} matrix which enables them to explore the effect of changes in network structure on firms' profits and consumer welfare. These changes include entering of a firm in a new market and also merging of two firms. Their results challenge the standard beliefs in Cournot oligopoly that more competition necessarily leads to higher welfare. Relatedly, Abolhassani et al. \cite{mamtagh} turn their focus on finding algorithms that compute pure Nash equilibria in Cournot competitions in networks. In \cite{bertrand}, the impact of monopolies on social welfare in a certain model of Bertrand network competition is studied.

Another thread of studies relevant to ours is the ones analyzing interactions in the networks and their impact on aggregate measures \cite{jackson2010social,aliz, sajadi2018affective,zhang2008clustering}. Most relatedly, Acemoglu et al. \cite{aliz} have proposed a framework that paves the way to examine equilibria in such inter-dependent agents setting and discover the effect of small microeconomic shocks on the economy's aggregate performance. Acting as our main inspiring study, they use Taylor expansion to acquire insights on the impact of shocks. Their examinations yield different results about the economy's \textit{ex ante} aggregate performance in the case of linear and non-linear worlds. To understand how the structure of the network shapes the economy's performance, they demonstrate that the \textit{Bonacich centrality} measure can capture this effect when the nature of the interactions is linear. Such analysis is prevalent in economics. For example, the general notion of production networks demands consideration of dependencies and network effects \cite{primer}.

Lastly, following the connections found between \textit{Bonacich centrality} and network effects in network Cournot competition \cite{shayan,ilkilic2009cournot}, studies about controlling centrality measures in networks can be considered related to ours. Generally, with an established relationship between centrality measures and social welfare in our setting, one might use these methods to change the structure of the competition such that social welfare increases. As such articles, \cite{nicosia2012controlling,liu2012control,bonacich1998controlling} model the centrality control problem as an optimization problem and presents an algorithm to solve it.

\section{Problem Formulation}
Consider a netwok game $\mathcal{G}$ which consists of $n$ firms $F = \{f_1, \cdots, f_n\}$ and $m$ markets $M = \{m_1, \cdots, m_m\}$ in which the firms compete. Each firm has access to a set of markets, meaning that it can supply the goods only in those specific markets. For firm $f_i$, let $M_i$ be the set of those markets. 

Similarly, let $F_j$ denote the set of firms that have access to market $m_j$. The amount of good that firm $f_i$ supplies in $m_j$ is shown by $q_{ij}$. Moreover, firm $f_i$ would incur the production cost $C_i(q)$ ($q$ is the vector of all $q_{ij}$s). We consider the inverse demand functions of the markets as linear. More specifically, market $m_j$ is governed by the relation 
\begin{equation}
P_j(q) = \alpha_j - \beta_j \sum_{f_i \in F_j} q_{ij}.
\end{equation}
 Additionally, 
 \begin{equation}
 C_i(q) = c_i\cdot(\sum_{m_j \in M_i} q_{ij})^2.
 \end{equation}
  For the sake of our analysis, we at first, suppose that for all $m_i \in M$, $\alpha_i = \alpha$ and $\beta_i = \beta$ and for all $f_j \in F$, $c_j = c$. We model this economy with a bipartite graph $G =(V,E)$. An example of this graph can be seen in Figure \ref{f:samplegraph}. 

\begin{figure}[h]
\centering
\includegraphics[width=0.3\textwidth]{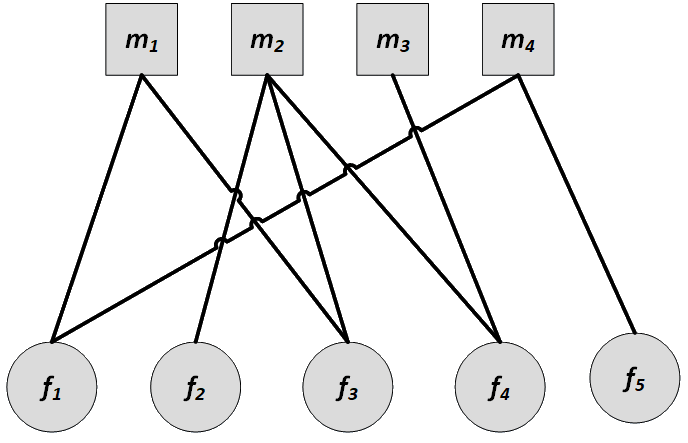}
\caption{A Graph for a Networked Market}
\label{f:samplegraph}
\end{figure}

Briefly speaking, we can consider each $f_i$'s profit as a combination of the afformentioned components:
\begin{equation}
\pi_i(q) = \sum_{m_k \in M_i}{q_{ik}\cdot P_k(q)}-C_i(q)
\end{equation}
 So given a network market graph $G$, each $f_i$ in competition with other firms solves the following optimization problems for computing its best response. 
\begin{equation}
\begin{array}{lll}{\underset{q_{i}}{\operatorname{maximize}}} & {\pi_{i}\left(\boldsymbol{q}_{i}, \boldsymbol{q}_{-\boldsymbol{i}}\right)} & {} \\ {\text { subject to }} & {q_{i k} \geq 0} & {\text { for } m_k \in M_{i}} \\ {} & {q_{i k}=0} & {\text { for } m_k \notin M_{i}}\end{array}
\end{equation}
where $q_i$ and $q_{-i}$ denotes the vector of production quantities of $f_i$ and its competitors, respectively. 

Bimpikis et al. \cite{shayan}, have focused on the equilibrium analysis of this model and their main result about the unique Nash equilibrium of this game is adopted as the foundation of this research.

\begin{theorem} \label{t:equil} [Adopted from Bimpikis et al. \cite{shayan}] The unique Nash equilibrium of the game is given by
$$q^* = [I+\gamma W]^{-1}\gamma \bar{\alpha},$$
where $\gamma = \frac{1}{2(c+\beta)}$, $\bar{\alpha}$ is a $|E| \times 1$ vector such that for every edge $(i,k)\in E$ we have $\bar{\alpha}_{ik}=\alpha_k$ and $W$ is an $|E|\times |E|$ matrix whose entries are
$$w_{i_{1} k_{1}, i_{2} k_{2}}=\left\{\begin{array}{l}{2 c \text { if } i_{1}=i_{2}, k_{1} \neq k_{2}} \\ {\beta \text { if } i_{1} \neq i_{2}, k_{1}=k_{2}} \\ {0 \text { otherwise }}\end{array}\right.$$

\end{theorem}

The matrix $[I+\gamma W]^{-1}$ is called the \textit{Leontief inverse}. We assume that the matrix $[I+\gamma W]$ is inversible. For a given realistic economy, this assumption is shown to be true \cite{sonis2000new,waugh1950inversion}. In this paper, we propose and formalize the problem of optimized governance of the afformentioned network market (networked cournot competition) with the objective of maximizing the social welfare. Social welfare ($SW$) in cournot competitions is defined as the sum of consumer surplus ($CS$) and firms' total profit. The consumer surplus in the Nash equilibrium is computed by the following formula (see \cite{hammond2002essential,shayan}):
\begin{equation}
CS =  \sum_{m_k \in M} \frac{\left(\alpha_{k}-P_{k}(q^*)\right)^{2}}{2 \beta}.
\end{equation}
Therefore the social welfare's formula is:
\begin{equation}
\begin{array}{ll}
SW  & =\sum_{f_i \in F}{\pi_i(q^*)} + CS \\
{ }  & = \sum_{f_i \in F}{[\sum_{m_k \in M_i}{q^*_{ik}\cdot P_k(q^*)}-C_i(q^*)]}\\
{ } & \hspace{0.5cm}+ \sum_{m_k \in M} 2/\beta \cdot \left(\alpha_{k}-P_{k}(q^*)\right)^{2}
\end{array} 
\end{equation}

Now, assume that an additional node which is called the government and has access to all markets is added to the network. This node's target is to maximize the social welfare without concerning its own utility. It has been provided with a budget of $B$ and can use this budget to provide some shocks to each market. In this paper, each shock to market $m_k$ is defined as provisioning some goods ($\epsilon_k \leq q^t$) to this market alongside the competing firms. $q^t$ is a threshold that is forced by external entities such as law or social pressure. Firms compete until they reach an equilibrium. The equilibrium can be computed by the following theorem.

\begin{theorem} \label{equil2}
The unique Nash equilibrium of the game $\mathcal{G}$ in the presence of shocks $\{\epsilon_i\}_{i=1}^{m}$ is given by
$$q^{*}=[I+\gamma W]^{-1} \gamma(\bar{\alpha}-\overline{\beta \epsilon}),$$
where $W$ and $\overline{\alpha}$ are defined as in Theorem \ref{t:equil} and $\overline{\beta \epsilon}$ is a $|E| \times 1$ vector such that for every edge $(i,k)\in E$ we have $\bar{\beta \epsilon}_{ik}=\beta_k \cdot \epsilon_k$
\end{theorem}
\begin{proof}
In the presence of shocks we can rewrite the firms' utility functions:
$$
\begin{array}{ll}
\pi_i(q, \epsilon) & = \sum_{m_{k} \in M_{i}} q_{i k} \cdot P_{k}(q)-C_{i}(q) \\
{ } & = \sum_{m_k \in M_i}{q_{ik}[\alpha_k-\beta_k(\sum_{f_j\in F_k}{q_{jk}} + \epsilon_k)]} \\
{ } & \hspace{0.5cm} - c_i\cdot(\sum_{m_k\in M_i}{q_{ik}})^2 \\
{ } & = \sum_{m_k \in M_i}{q_{ik}[(\alpha_k-\beta_k\epsilon_k)+\beta_k\sum_{f_j\in F_k}{q_{jk}}]} \\
{ } & \hspace{0.5cm} - c_i\cdot(\sum_{m_k\in M_i}{q_{ik}})^2 \\
\end{array}.
$$

Assume that we have a new game $\mathcal{G'}$ where everything is the same as the previous setting ($\mathcal{G})$ except that the values of $\alpha_k$s have changed to $\alpha_k - \beta_k\epsilon_k$. By Theorem \ref{t:equil} , we will have the following formula for the Nash equilibrium point:
$$q^{*}=[I+\gamma W]^{-1} \gamma(\bar{\alpha}- \overline{\beta \epsilon})$$

Equilibria in $\mathcal{G}$ in the presence of shocks are equal to equilibria of $\mathcal{G'}$ because $\pi_i$s computed by the above formula is exactly what must be for $\mathcal{G'}$. So by following the method used in \cite{shayan} for proving Theorem \ref{t:equil}, our desired target will be achieved.

\end{proof} 
Note that for each $m_k \in M$, we must have $\epsilon_k < \alpha / \beta$, because if not, the price function $P_k$ will be negative and it is not acceptable. Now we can write the formulation of social welfare. The $SW$ function in the presence of shocks can be recalculated as follows:

$$
\begin{array}{ll}
SW  & =\sum_{f_i \in F}{\pi_i(q^*,\epsilon)} + CS \\
{ }  & = \sum_{f_i \in F}{\big[\sum_{m_k \in M_i}{q^*_{ik}\cdot P_k(q^*, \epsilon)}-C_i(q^*)\big]}\\
{ } & \hspace{0.5cm}+ \sum_{m_k \in M} 2/\beta_k \cdot \big(\alpha_{k}-P_{k}(q^*,\epsilon)\big)^{2} \\
{ } & = \sum_{f_i \in F}\Big[\sum_{m_k \in M_{i}} q^*_{i k}\big(\alpha_{k}-\beta \sum_{f_j \in F_{k}} q^*_{j k}-\beta\epsilon\big) \\
{ } & \hspace{1.5cm} -c\big(\sum_{m_k \in M_{i}} q^*_{i k}\big)^{2}\Big]\\
{ } & \hspace{0.5cm} + \sum_{m_k \in M} 2/\beta \cdot \Big[\big(\beta\sum_{f_j \in F_k}{q^*_{jk}+\beta\epsilon}\big)^2\Big] \\
{ } & = \sum_{m_k \in M} \Big(\sum_{f_i \in F_{k}} q^*_{i k} \alpha_{k}\Big) \\
{ } & \hspace{0.5cm}-(\beta/2)\cdot \sum_{m_k\in M}\Big(\sum_{f_i \in F_{k}} q^*_{i k}\Big)^{2}\\
{ } & \hspace{0.5cm}-c \sum_{f_i\in F}\Big(\sum_{m_k \in M_{i}} q^*_{i k}\Big)^{2} \\
{ } & \hspace{0.5cm}-\sum_{m_k \in M}\Big((\beta/2)\cdot \epsilon_k^2+\sum_{f_i \in F_k}q^*_{ik}\epsilon_k \Big) 
\end{array}
$$
Therefore, we have the following vectorized formula, 
\begin{equation}\label{sw}
\begin{aligned} 
SW  = {q^*}^T\overline{\alpha}-(\beta/2+c){q^*}^T q^*- (1/2){q^*}^TW q^*\\
-\overline{\beta \epsilon}^Tq^* - (\beta/2)\overline{\epsilon}^T\overline{\epsilon}, \hspace{3cm}
\end{aligned}
\end{equation}

where $\overline{\epsilon}$ is a $m\times 1$ vector whose $k$th component is equal to $\epsilon_k$ and other things are defined as before (see Theorem \ref{t:equil} and Theorem \ref{equil2}).
By the above formulation the problem of governing social welfare with market shocks can be modeled by the optimization problem described in Definition \ref{finalOPT}. 

\begin{definition}  \label{finalOPT}
The problem of governing (maximizing) social welfare with shocks in a networked market $\mathcal{G}$ ($MaxSW(\mathcal{G})$) is defined as the following optimization problem:
\begin{equation*}
\begin{array}{ll@{}ll}
\text{Maximize}  & {q^*}^T\overline{\alpha}-(\beta/2+c){q^*}^T q^*- (1/2){q^*}^TW q^*&\\
{ }   &\hspace{0.5cm}-\overline{\beta \epsilon}^Tq^* - (\beta/2)\overline{\epsilon}^T\overline{\epsilon}& \\
\text{subject to}& \displaystyle  q^{*}=[I+\gamma W]^{-1} \gamma(\bar{\alpha}-\overline{\beta \epsilon})  & \\
                 &  \displaystyle  c\cdot(\sum_{m_k \in M} \epsilon_k)^2 \leq B  & { } \\
                 &  \displaystyle  0 \leq \epsilon_k \leq q^t  \hspace{2.5cm} \forall m_k \in M &\\
\end{array}
\end{equation*}
\end{definition}

All the elements of the sum depicted in the $SW$ formula are concave except $-\overline{\beta \epsilon}^Tq^* $ which makes the convex optimization frameworks unusable. In the next section, we devise a heuristic algorithm for this optimization problem. This is done by proposing a linear estimation for the social welfare and an optimization algorithm for maximizing it. Then in Section \ref{experiments}, by experimentation on real and synthetic data, the good performance of this heuristic algorithm is shown.

\section{Solution Estimation}
In this section, we provide some insights about the social welfare function and by linearizing it with Taylor expansion, we propose an algorithm called the \textit{Linear} heuristic for the $MaxSW(\mathcal{G})$ problem. More precisely in this section, we propose a metric that can be computed by the network structure and we analytically show that picking the markets with larger amounts of it can (approximately) maximize the social welfare. In the next section by running experiments on real and synthetic datasets, we will demonstrate the superiority of this metric over other metrics.

The main idea for this aim is to use the first order multivariate Taylor expansion \cite{konigsberger2013analysis} to create a linear approximation for the social welfare function. This approximation leads to a linear combination of $\epsilon_i$s:
$$SW(\epsilon) = SW(0) + \zeta_1 \epsilon_1+\zeta_2\epsilon_2+...+\zeta_m \epsilon_m.$$
Keeping in mind that the government has a limited budget for its interventions, we can consider the shocks small ($\forall_{m_k \in M}\epsilon_k \leq q^t$), so this approximation may be valid. The coefficient of each $\epsilon_i$ ($\zeta_i$) can be considered as the afformentioned metric. Since $SW$ is a differentiable function, we can write $SW$ as follows:
\begin{equation}
\begin{aligned}
SW(\epsilon) \approx SW(0)+\epsilon \cdot \nabla SW(0) \hspace{0.45cm}\\
= SW(0) + \sum_{r = 1}^m{\epsilon_r \frac{\partial SW}{\partial \epsilon_r}|_{\epsilon=0}}
\end{aligned}
\end{equation}
So, the amount of social welfare added by shocks is a linear combination of $\epsilon_r$s whose coefficients are $\zeta_r = \frac{\partial SW}{\partial \epsilon_r}|_{\epsilon=0}$. So to maximize the amount of social welfare, markets should be supplied in order of their $\zeta_r$s. The structure of this algorithm can be seen in Algorithm \ref{algo}. Here's how to calculate the coefficients.

Using Theorem \ref{equil2} and expanding the formula derived for $q^*$, we have:
\begin{equation}
\begin{aligned} 
q_{i k}^{*}=\frac{\alpha_{k}-2 c \sum_{m_\ell \in M_{i}, m_\ell \neq m_k} q_{i \ell}^{*}-\beta \sum_{f_j \in F_{k}} q_{j k}^{*}}{2(\beta+c)}\\
\hspace{-0.6cm}=\frac{\alpha_{k}}{2(\beta+c)}-\sum_{(j,\ell) \in E\left(\mathcal{G}\right)}(\gamma W)_{i k, j \ell} q_{j \ell}^{*} \hspace{1.5cm}
\end{aligned}
\end{equation}

If we define function $f(z) = \gamma \alpha - \gamma z$, 
\begin{equation}
q^*_{ij} = f(\sum_{k, \, l} w_{ij, k\ell} q^*_{k\ell} + \gamma \beta \epsilon_j)
\end{equation}
 denotes the amount firm $f_i$ supplies in market $m_j$ in equilibrium under the presence of shock $\epsilon_j$. Moreover, from Equation \ref{sw} with $h(x) = \alpha x - (\frac{\beta}{2} + c) x^2$ and $u(q_{ij}, q_{k\ell}) = w_{ij, k\ell} q_{ij} q_{k\ell}$,
 \begin{equation}
\begin{aligned} 
SW = \sum_{i, \, j} h(q^*_{ij}) + \sum_{m_k \in M} h(\epsilon_k)\hspace{2cm}\\
 - \frac{1}{2} \sum_{i, \, j, \, k \, , \ell} u(q^*_{ij}, q^*_{k\ell}) - \beta \sum_{k=1}^{m} \epsilon_k \sum_j q^*_{jk} 
\end{aligned}
\end{equation}
is the social welfare under the circumstances discussed so far.

We start our analysis by calculating $\frac{\partial q^*_{ij}}{\partial \epsilon_r}$. By the idea used in \cite{aliz}, 
\begin{equation}
\begin{aligned}
\frac{\partial q^*_{ij}}{\partial \epsilon_{r}}=f^{\prime}\left(\sum_{i, \, j, \, k \, , \ell} w_{ij, k\ell} q^*_{k\ell} + \gamma \beta \epsilon_{j}\right)\cdot \hspace{2cm}\\
\hspace{2cm} \left(\sum_{i, \, j, \, k \, , \ell} w_{ij, k\ell} \frac{\partial q^*_{k\ell}}{\partial \epsilon_{r}}+\gamma \beta \mathbf{1}\{r=j\}\right) 
\end{aligned}
\end{equation}
Evaluating this equation using the matrix form at the point
 $\epsilon = (\epsilon_1,\cdots, \epsilon_{|E|}) = 0$, which is the absence of the government, 
yields
\begin{equation}
\begin{aligned}
\frac{\partial q^*}{\partial \epsilon_r}|_{\epsilon=0} = -\gamma \beta (I + \gamma W)^{-1} e_r\text{,}
\end{aligned}
\end{equation}
where $e_r$ is a $|E| \times 1$ vector, with ones for edges connecting to market $m_r$ and zeros elsewhere. Thus, we have:

\begin{equation}
\begin{aligned}
\frac{\partial q^*_{ij}}{\partial \epsilon_r}|_{\epsilon=0} = - \gamma \beta \sum_{k} \lambda_{ij, kr}\text{ ,} \label{eq:1}
\end{aligned}
\end{equation}
where $\lambda_{ij, kr}$ is the corresponding element to edges $ij$ and $kr$ of matrix $(I + \gamma W)^{-1}$.

Setting our sights on social welfare, we have:

\begin{equation}
\begin{aligned}
\frac{\partial \text{ SW}}{\partial \epsilon_{r}}= \sum_{i, j} h^{\prime}\left(q^*_{ij}\right) \frac{\partial q^*_{ij}}{\partial \epsilon_{r}}  + \alpha - (\beta + 2c) \epsilon_r \hspace{1.6cm} \\
- \frac{1}{2}\sum_{i, j, k, \ell} \left[\frac{\partial u}{\partial q^*_{ij}} \frac{\partial q^*_{ij}}{\partial \epsilon_r} + \frac{\partial u}{\partial q_{k\ell}} \frac{\partial q^*_{k\ell}}{\partial \epsilon_r}\right] - \beta \sum_k \epsilon_k \sum_j \frac{\partial q^*_{jk}}{\partial \epsilon_r}.
\end{aligned}
\end{equation}

Considering $h^\prime(x) = \alpha - (\beta + 2c)x$ and $\frac{\partial u(q_{ij}, q_{k\ell})}{\partial q_{ij}} = w_{ij, k\ell} q_{k\ell}$, we get:

\begin{equation} \label{eq:2}
\begin{aligned}
\frac{\partial \text{ SW}}{\partial \epsilon_{r}}= \sum_{i,j} (\alpha - (\beta + 2c) q^*_{ij}) \frac{\partial q^*_{ij}}{\partial \epsilon_{r}} +\hspace{2cm}  \\
 - \frac{1}{2}\sum_{i, j, k, \ell} \left[w_{ij, k\ell} q^*_{k\ell} \frac{\partial q^*_{ij}}{\partial \epsilon_r} + w_{ij, k\ell} q^*_{ij} \frac{\partial q^*_{k\ell}}{\partial \epsilon_r}\right] \\
 - \beta \sum_k \epsilon_k \sum_j \frac{\partial q^*_{jk}}{\partial \epsilon_r}+ \alpha - (\beta + 2c) \epsilon_r) \hspace{0.6cm}
\end{aligned}
\end{equation}
Thus, we evaluate equation \eqref{eq:2} at $\epsilon = 0$:

\begin{equation}\label{eq:3}
\begin{aligned}
\frac{\partial \text{ SW}}{\partial \epsilon_{r}}|_{\epsilon=0} = \sum_{i,j} (\alpha - (\beta + 2c) q^*_{ij}|_{\epsilon=0}) \frac{\partial q^*_{ij}}{\partial \epsilon_{r}}|_{\epsilon=0} \\
 - \frac{1}{2}\sum_{i, j, k, \ell} \Big[w_{ij, k\ell} q^*_{k\ell}|_{\epsilon=0} \frac{\partial q^*_{ij}}{\partial \epsilon_r}|_{\epsilon=0} \hspace{0.5cm}\\+ w_{ij, k\ell} q^*_{ij}|_{\epsilon=0} \frac{\partial q^*_{k\ell}}{\partial \epsilon_r}|_{\epsilon=0}\Big] + \alpha ) \hspace{0.8cm}
\end{aligned}
\end{equation}

$q^*_{ij} |_{\epsilon=0}$ has been studied in \cite{shayan}. Based on their results, we can deduce the following in our setting:

\begin{equation}
\begin{aligned}
q^*_{ij} |_{\epsilon=0} = \gamma \alpha \sum_{k\ell} \lambda_{ij, k\ell}\text{ .}\label{eq:5}
\end{aligned}
\end{equation}

Using equation \eqref{eq:1} and \eqref{eq:5}, we get:

\begin{equation}
\begin{aligned}
\frac{\partial \text{ SW}}{\partial \epsilon_{r}}|_{\epsilon=0}= -\gamma \beta\sum_{ij} (\alpha - \gamma \alpha(\beta + 2c) \sum_{k\ell} \lambda_{ij, k\ell}) \sum_{k} \lambda_{ij, kr} - \\ \sum_{i,j} (\gamma \alpha \sum_{k\ell} \lambda_{ij, k\ell} (\sum_{k\ell} w_{ij, k\ell} \sum_{t} \gamma \beta \lambda_{k\ell, tr})) + \alpha \label{eq:6}
\end{aligned}
\end{equation}

Since we have derived a formula for computing $\zeta_r = \frac{\partial SW}{\partial \epsilon_r}|_{\epsilon=0}$s, we can state the final algorithm. This algorithm is shown in Algorithm \ref{algo}. In the first four lines of this algorithm $\zeta_r$s are computed from the network structure and Leontief matrix. After that a set $T$ is initialized to the set of all markets and a variable $S$ is initialized to 0. In the next while loop, $T$ is to be the set of markets that have not been supplied with shocks so far and $S$ is defined as the total goods supplied by the government. So the loop execution will continue until either all markets are supplied with shocks or the cost of shock supplies exceeds the budget intended for this purpose ($B$).

In each iteration of the while loop, the market with maximum $\zeta_r$ is chosen from $T$ and extracted from this set. Then the maximum possible shock is computed as $\epsilon_r$ and is added to $S$. Note that for computing $\epsilon_r$, three upper bounds must be considered:
\begin{enumerate}
\item $q^t$ is the upper bound defined in Definition \ref{finalOPT}. 
\item $\frac{\beta}{\alpha}$ is the upper bound defined by the price function. If $\epsilon_r > \frac{\beta}{\alpha}$, the price will be negative at that market and it is acceptable. 
\item $\sqrt{\frac{B}{c_r}}-S$ is the amount of goods that can be provided by the remaining budget. 
\end{enumerate}

\begin{algorithm*}
\DontPrintSemicolon 
\KwIn{A network market $\mathcal{G}$ alongside with parameters $\alpha$, $\beta$ and $\gamma$}
\KwOut{The amount of shocks $\epsilon_1$, $\epsilon_2$, ..., $\epsilon_m$ which makes the maximum social welfare}
\;

Set $\lambda_{ij,kr}$ equal to the corresponding elements to edges $ij$ and $kr$ of matrix $(I+\gamma W)^{-1}$\;
\For{$r \gets 1$ \textbf{to} $m$} {
   $\zeta_r \gets -\gamma \beta\sum_{ij} (\alpha - \gamma \alpha(\beta + 2c) \sum_{k\ell} \lambda_{ij, k\ell}) \sum_{k} \lambda_{ij, kr} -  \sum_{i,j} (\gamma \alpha \sum_{k\ell} \lambda_{ij, k\ell} (\sum_{k\ell} w_{ij, k\ell} \sum_{t} \gamma \beta \lambda_{k\ell, tr})) + \alpha$
}
$T \gets M$
$S \gets 0$\;
\While{$T \neq \emptyset$ and $c\cdot S^2 < B$}{
Set $r$ to the index of the market $m_r \in T$ with maximum $\zeta_r$ \;
$T \gets T \setminus m_r$ \;
$\epsilon_r \gets \min\{q^t, \frac{\beta_r}{\alpha_r}, \sqrt{\frac{B}{c_r}}-S \}$\;
$S \gets S + \epsilon_r$
}
\Return{$\epsilon_r$s} \;
\caption{The Linear Heuristic for solving $MaxSW(\mathcal{G})$ }
\label{algo}
\end{algorithm*}

\section{Empirical Study} \label{experiments}

We evaluate the performance of our proposed measure and method on the synthetic data and a real-world dataset of different pharmaceutical companies as our firms and different markets for drugs. For example, Aspirin and its users define a market in which players are companies that produce this drug. We collect this dataset from different 135 drug companies that produce 603 drugs altogether. Additionally, we use identical parameters $\alpha, \beta, c$ for all the firms and markets, as we are considering the symmetric case. These parameters are set in a way to be close to real-life values. The synthetic data also has 603 markets and 135 firms, which has been chosen randomly. The characteristics of this dataset are shown in Table \ref{dataset}. A subgraph of this network is shown in Figure \ref{subgraph}.  
\begin{figure}[h]
\centering
\includegraphics[width=0.3\textwidth]{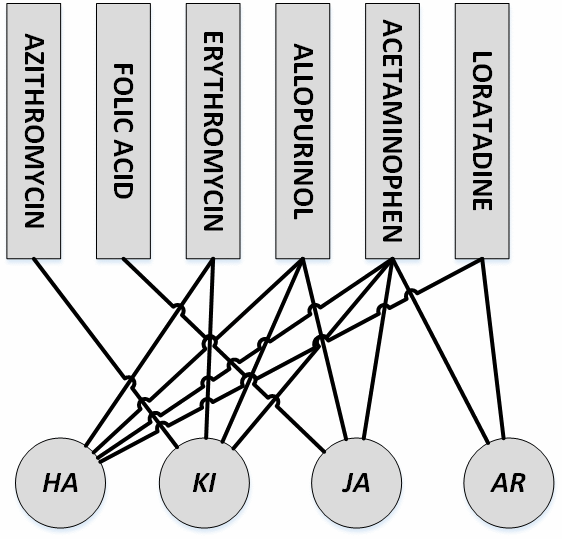}
\caption{A Subgraph of the Drug Company Dataset}
\label{subgraph}
\end{figure}

\begin{table} \centering 
\caption{Summary of Dataset's Characteristics}
\begin{tabular}{ |p{4cm}|p{2cm}|  } 
\hline 
 \multicolumn{2}{|c|}{Drug Companies Dataset} \\
 \hline
Characteristic& Value \\
 \hline
 \#Markets   & 603    \\
 \#Firms&   135  \\
 \#Firm-Market pairs (edges) & 2049 \\
 \hline
\end{tabular}
\label{dataset}
\end{table}

\subsection{Competitor Benchmark}

The essence of competitors we consider is that the government takes on a measure to rank the markets. Next, it supplies goods to markets in that order, as much as possible and as long as permissible. Naively, it is possible to choose the markets at random. No measure, to the best of our knowledge, has been presented for picking the markets yet. Nevertheless, centrality measures are natural candidates for us to use as benchmarks. As for the centrality measures we consider, we use the followings:

\begin{itemize}

\item \textbf{Degree.} 
The simplest centrality measure is the degree of a node, which in our model, is the number of firms competing in a market:
$$
d(m_i) = |F_i|.
$$

\item \textbf{Betweenness.}
Generally speaking, betweenness centrality is a quantity for determining the impact of a node over the flow of information in a graph \cite{brandes2001faster}. The betweenness of a node is an indicator for the fraction of shortest paths in a graph that pass through this vertex. In our setting:

$$
b(m_i)=\sum_{f_j, f_k \in F} \frac{\sigma_{j k}(m_i)}{\sigma_{j k}} ,
$$
where $\sigma_{j k}$ is the total number of shortest paths from firm $f_j$ to firm $f_k$ and $\sigma_{j k}(m_i)$ is the total number of those paths that include market $m_i$.

\item \textbf{Closeness.}
Closeness centrality is an aggregate measure of a node's proximity to other nodes. More precisely, closeness of node $v$ is defined as the inverse of sum of distances of node $v$ from other nodes. Considering our model:

$$
cl(m_i) = \sum_{f_j\in F} \frac{1}{dist(m_i, f_j)} ,
$$
where $dist(m_i, f_j)$ is the distance between market $m_i$ and firm $f_j$ in the graph.

\end{itemize}

Regardless of the measure one picks, it is possible to adopt the ascending or the descending order. Thus, we consider both choices, however, the ascending order empirically shows better performance with these benchmark centrality measures. In conclusion, we present our studied strategies in Table \ref{heuristics}. 

\begin{table}\centering
\caption{Summary of proposed algorithms in the competitor benchmark}
\begin{tabular}{@{}cccl@{}}\toprule
& \multicolumn{1}{c}{Heuristic} & \phantom{abc} & \multicolumn{1}{c}{Description}\\ \midrule
& $Linear$ && Descending order of $\zeta_i$\\
& $AscDeg$ && Ascending order of $d(m_i)$\\
& $DescDeg$ && Descending order of $d(m_i)$\\
& $AscBet$ && Ascending order of $b(m_i)$\\
& $DescBet$ && Descending order of $b(m_i)$\\
& $AscCL$ && Ascending order of $cl(m_i)$\\
& $DescCL$ && Descending order of $cl(m_i)$\\
& $Random$ && Random order of $m_i$\\
 \bottomrule
\end{tabular}
\label{heuristics}
\end{table}

After having our graph, we use Theorem \ref{equil2} to calculate the equilibrium of the game with our parameters $\alpha$, $\beta$, and $c$. Subsequently, we pick a policy $A$ for governance of social welfare from Table \ref{heuristics}. Then, the government begins supplying the amount $q^t$ of the commodity into the markets in the corresponding order, until the supplies violate the constraints in Definition \ref{finalOPT}. Thus, the government would supply up to $\left\lfloor \sqrt{\frac{B}{C}} \right\rfloor$, where $B$ is the budget that innates to the $MaxSW(\mathcal{G})$ problem. Altering $B$ gives us a trajectory for social welfare. Let $SW_{A}(B)$ be the social welfare obtained by applying strategy $A$ on $MaxSW(\mathcal{G})$ with parameter $B$. In the next subsection, we compare trajectories $SW_{A}(B)$ for policies in Table \ref{heuristics}. 

\subsection{General performance}

Our empirical results are indicated in Figures \ref{f1}, \ref{f2}, \ref{f3}, \ref{f4}, \ref{f5} and \ref{f6}. For each policy $A$ of Table \ref{heuristics}, we plot the difference between the amount of social welfare which can be obtained by our Linear algorithm and the social welfare which can be obtained by the heuristic $A$, i.e. $SW_{Linear}(B) - SW_{A}(B)$.  

We observe that our proposed measure is strictly better than other mentioned quantities. For the random case, we use the average results over 50 different realizations to extract the expected performance. It is good to note that picking the markets according to ascending order of a centrality measure seems to be better than the reverse order. This is intuitive, since markets with lower centralities are generally more monopolized and government's interventions have more impact on them. This observation is in accordance with the general belief in economics that higher competition leads to higher social welfare \cite{}. In addition, the difference in social welfare obtained by the Linear heuristics and other methods tend to rise, until a certain point at which drops. Since the end point of the plots suggests the government to supply in almost all markets by each method, the last points have $y$ coordinates close to zero. All the plots are strictly above the horizontal line $y = 0$, except beginning points which indicates the superiority of our proposed approach over other mentioned strategies.

\begin{figure} \centering
   \includegraphics[scale=0.13]{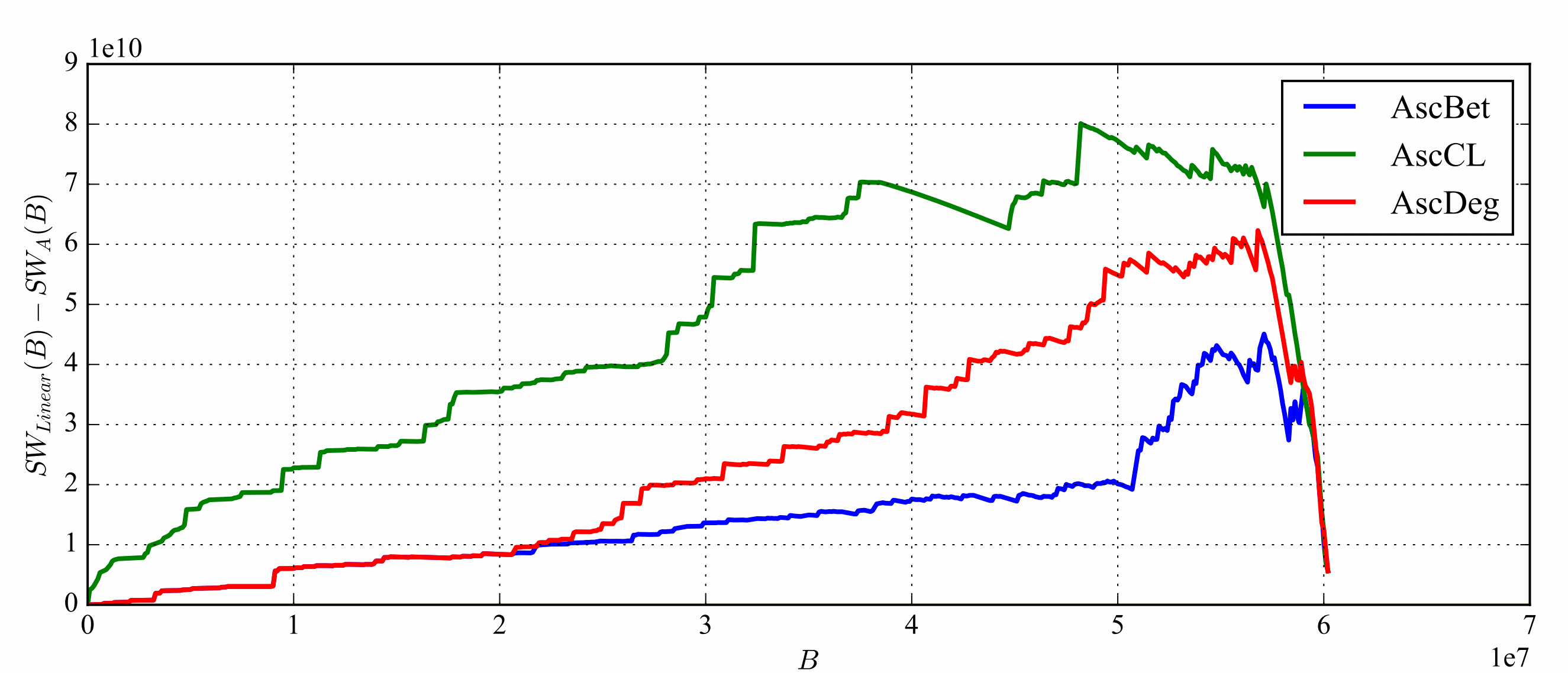}
   \caption{The difference between the social welfare obtained by the Linear heuristic and competitor structural heuristics with ascending order in the dataset of drug companies}
   \label{f1}
\end{figure}
\begin{figure} \centering
   \includegraphics[scale=0.13]{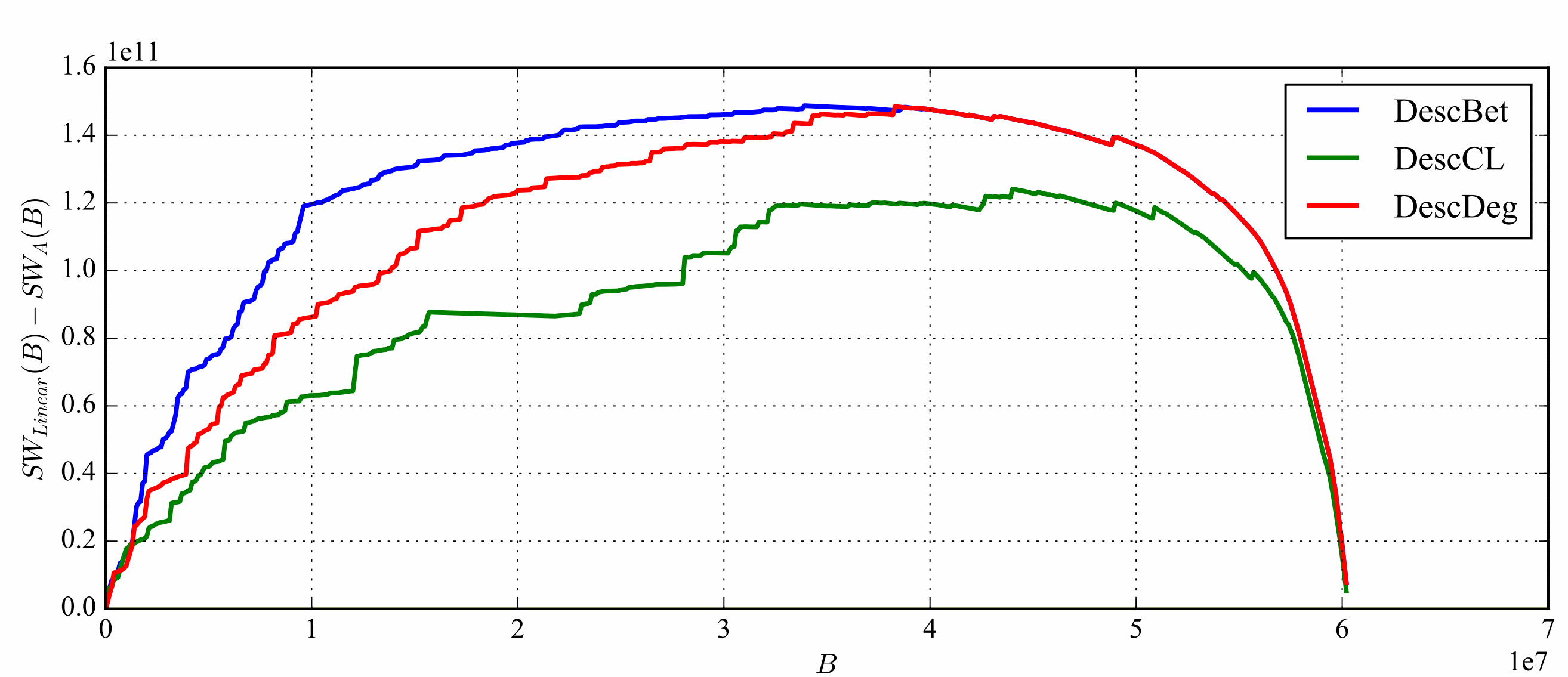}
   \caption{The difference between the social welfare obtained by the Linear heuristic and  competitor structural heuristics with descending order in the dataset of drug companies}
   \label{f2}
\end{figure}
\begin{figure} \centering
   \includegraphics[scale=0.13]{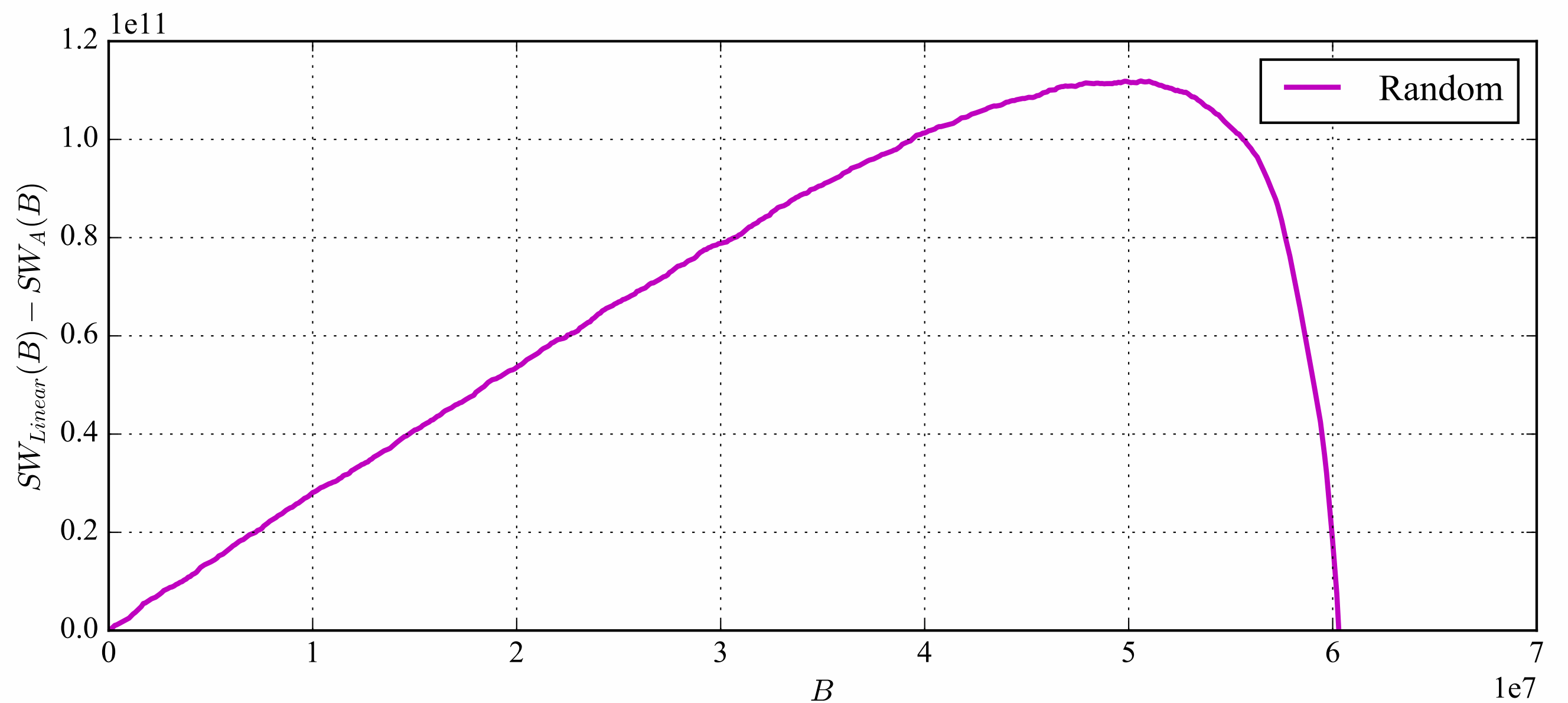}
   \caption{The difference between the social welfare obtained by the Linear heuristic and the random heuristic in the dataset of drug companies}
   \label{f3}
\end{figure}
\begin{figure} \centering
   \includegraphics[scale=0.13]{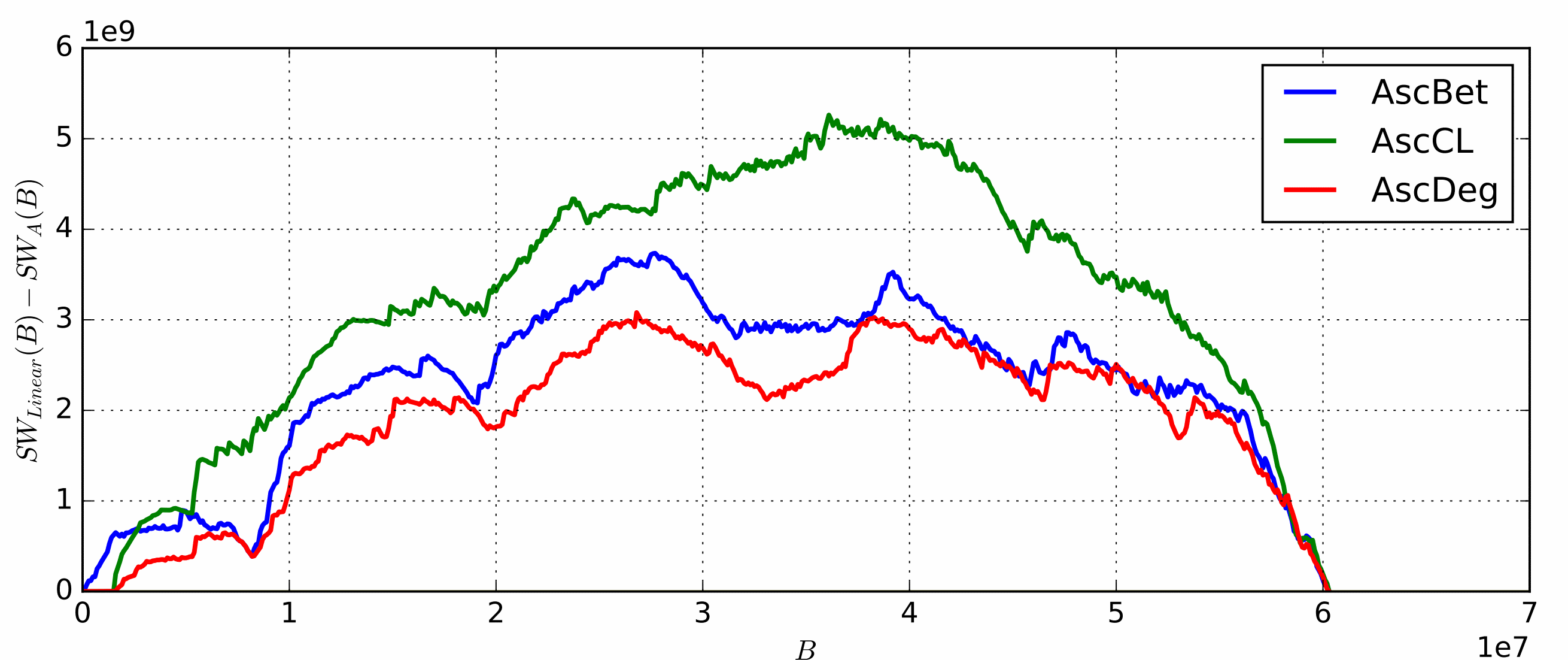}
   \caption{The difference between the social welfare obtained by the Linear heuristic and competitor structural heuristics with ascending order in the synthetic dataset}
   \label{f4}
\end{figure}
\begin{figure} \centering
   \includegraphics[scale=0.13]{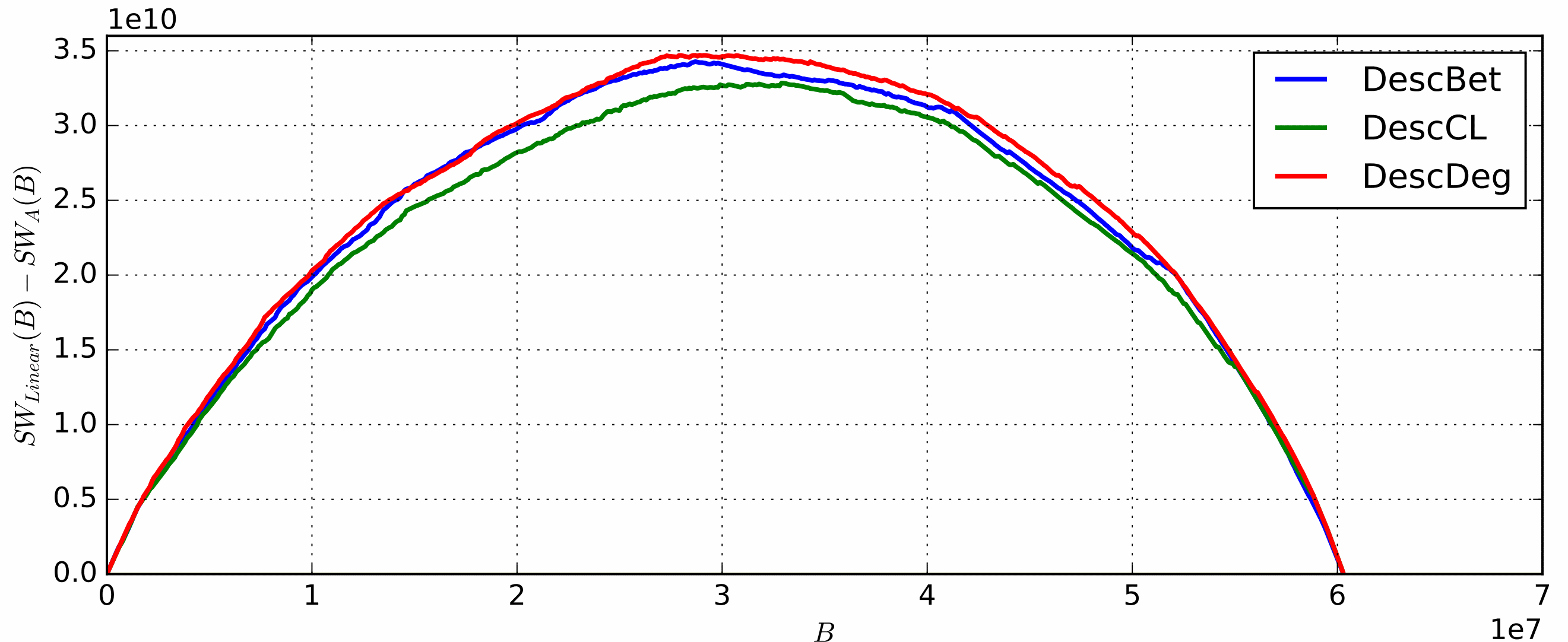}
   \caption{The difference between the social welfare obtained by the Linear heuristic and  competitor structural heuristics with descending order in the synthetic dataset}
   \label{f5}
\end{figure}
\begin{figure} \centering
   \includegraphics[scale=0.13]{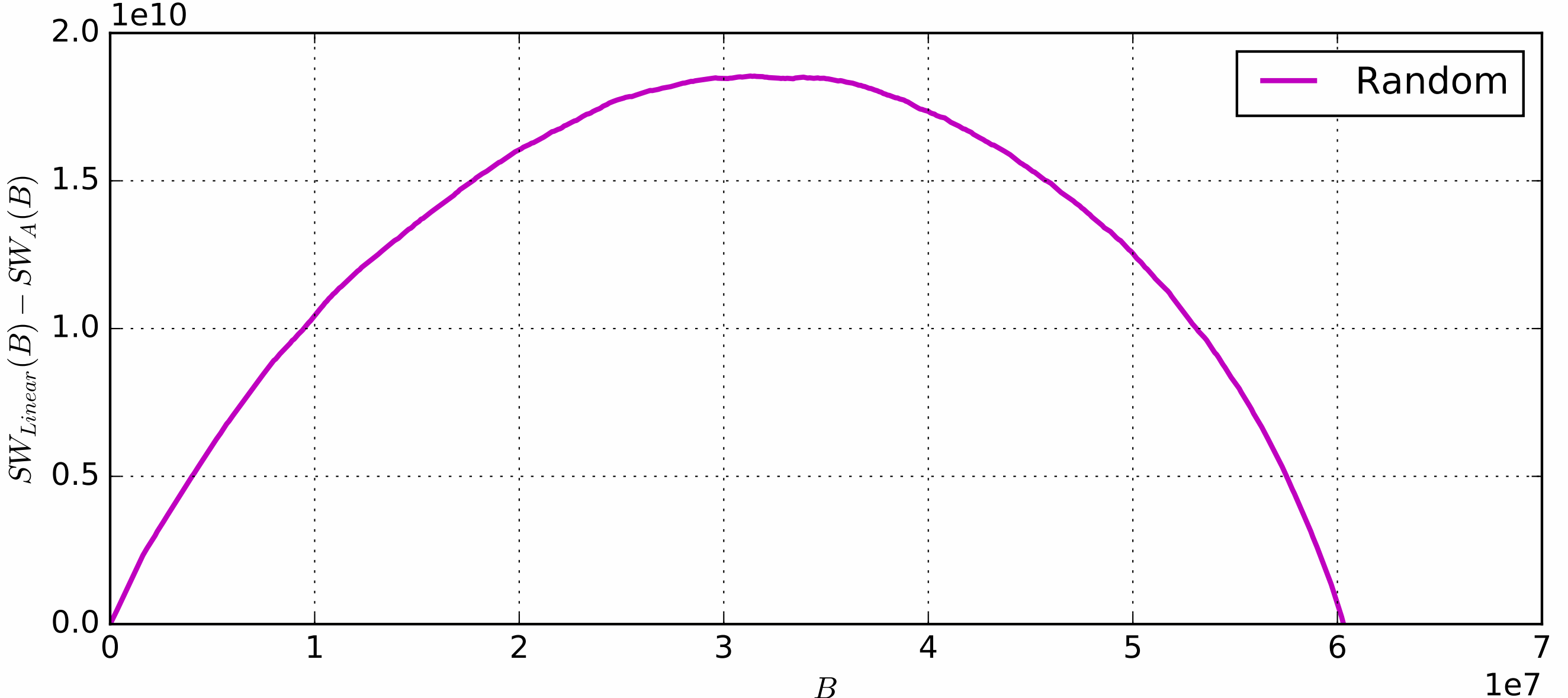}
   \caption{The difference between the social welfare obtained by the Linear heuristic and the random heuristic in the synthetic dataset}
   \label{f6}
\end{figure}

\clearpage

\bibliographystyle{unsrt}  

\end{document}